%% file: mis_arxiv_main.tex
\documentclass[11pt]{article}

\usepackage{mathptmx}

\input{settings-custom.tex}

\date{}

\title{Relaxed Schedulers Can Efficiently Parallelize Iterative Algorithms}
\author{Dan Alistarh \\IST Austria \and  Trevor Brown \\IST Austria \and Justin Kopinsky \\MIT \and Giorgi Nadiradze \\ETH Zurich}
\date{}

\begin{document}

\maketitle

\begin{abstract}
There has been significant progress in
understanding the parallelism inherent to iterative sequential algorithms: for
many classic algorithms, the depth of the dependence structure is now well
understood, and scheduling techniques have been developed to exploit this shallow dependence
structure for efficient parallel implementations. A related, applied
research strand has studied methods by which certain iterative task-based algorithms can be
efficiently parallelized via relaxed concurrent priority schedulers. These allow
for high concurrency when inserting and removing tasks, at the cost of 
executing superfluous work due to the relaxed semantics of the scheduler.

In this work, we take a step towards unifying these two research directions, by
showing that there exists a family of relaxed priority schedulers that can
efficiently and deterministically execute classic iterative algorithms such as
greedy maximal independent set (MIS) and matching. Our primary result shows that,
given a randomized scheduler with an expected relaxation factor of $k$ in terms of
the maximum allowed priority inversions on a task, and any graph on $n$ vertices, the scheduler
is able to execute greedy MIS with only an additive factor of $\poly(k)$ expected additional
iterations compared to an exact (but not scalable) scheduler. This counter-intuitive result demonstrates that the overhead of relaxation when computing MIS is \emph{not} dependent on the input size or structure of the input graph. 
Experimental results show that this overhead can be clearly offset by the gain
in performance due to the highly scalable scheduler. In sum, we present an
efficient method to deterministically parallelize iterative sequential
algorithms, with provable runtime guarantees in terms of the number of executed
tasks to completion.

\end{abstract}

\maketitle




\section{Introduction}

Given the now-pervasive nature of parallelism in computation, there has been a tremendous amount of research into efficient parallel algorithms for a wide range of tasks. 
A popular approach has been to  map existing sequential algorithms to parallel architectures, by exploiting their \emph{inherent parallelism}. 
In this paper, we will focus on two specific variants of this strategy. 

The \emph{deterministic} approach, e.g.~\cite{Blelloch, BFS12, Swarm, blelloch2012internally, BFS14, blelloch2016parallelism} has been to study the directed-acyclic graph (DAG) step dependence in  classic, widely-employed sequential algorithms, showing that, 
perhaps surprisingly, this dependence structure usually \emph{has low depth}. 
One can then design schedulers which  exploit this dependence structure for efficient execution on parallel architectures. 
As the name suggests, this approach ensures deterministic outputs (i.e. outputs uniquely determined by the input), 
and can yield good practical performance~\cite{BFS12}, but requires a non-trivial amount of knowledge about the problem at hand, and the use of carefully-constructed parallel schedulers~\cite{BFS12}. 

To illustrate, let us consider the classic sequential greedy  strategy for solving the maximal independent set (MIS) problem on arbitrary graphs: 
the algorithm examines the set of vertices in the graph following a fixed, random sequential priority order, adding a vertex to the independent set if and only if no neighbor of higher priority has already been added. 
The basic insight for parallelization is that the outcome at each node may only depend on a small subset of other nodes, namely its neighbors which are higher priority in the random order.
Blelloch, Fineman and Shun~\cite{BFS12} performed an in-depth study of the asymptotic properties of this dependence structure, proving that, for any graph, the maximal depth of a chain of dependences is in fact $O( \log^2 n )$ with high probability, where $n$ is the number of nodes in the graph. 
Recently, an impressive analytic result by Fischer and Noever~\cite{FN18} provided tight $\Theta( \log n )$ bounds on the maximal dependency depth for greedy sequential MIS, effectively closing this problem for MIS. 
Beyond greedy MIS, there has been significant progress in analyzing the dependency structure of other fundamental sequential algorithms, such as algorithms for matching~\cite{BFS12}, list contraction~\cite{BFS14},  
Knuth shuffle~\cite{BFS14}, linear programming~\cite{Blelloch}, and graph connectivity~\cite{Blelloch}.

An alternative approach has been to employ \emph{relaxed data structures} to schedule task-based programs. 
Starting with Karp and Zhang~\cite{KarZha93}, the general idea is that, in some applications, 
the scheduler can relax the strict order induced by following the sequential algorithm, and allow tasks to be processed speculatively ahead of their dependencies, without loss of correctness. 
A standard example is parallelizing Dijkstra's single-source shortest paths (SSSP) algorithm, e.g.~\cite{SprayList, making-kjell-happy, Nguyen13}: 
the scheduler can retrieve vertices in relaxed order without breaking correctness, as the distance at each vertex is guaranteed to eventually converge to the minimum.  
The trade-off is between the performance gains arising from using simpler, more scalable schedulers, and the loss of determinism and the wasted work due to relaxed priority order. 
This approach is quite popular in practice, as several high-performance relaxed schedulers have been proposed, which can attain state-of-the-art results in settings such as graph processing and machine learning~\cite{Nguyen13, gonzalez2012powergraph}. At the same time, despite good empirical performance, this approach still lacks analytical bounds, and results are no longer deterministic.

In this paper, we ask: is it possible to achieve both the simplicity and good performance of relaxed schedulers \emph{as well as} the predictable outputs and runtime upper bounds of the ``deterministic" approach? 

\paragraph{Contribution.} In a nutshell, this work shows that a natural family of \emph{fair relaxed schedulers}---providing upper bounds on the degree of relaxation, and on the number of inversions that a task can experience---can execute a range of iterative sequential algorithms \emph{deterministically}, preserving the dependence structure, and \emph{provably efficiently}, providing analytic upper bounds on the total work performed. 
Our results cover the classic greedy sequential graph algorithms for \emph{maximal independent set (MIS)}, \emph{matching}, and \emph{coloring}, but also algorithms for \emph{list contraction} and \emph{generating permutations} via the Knuth shuffle.  We call this class \emph{iterative algorithms with explicit dependencies}. 
Our main technical result is that, for MIS and matching in particular, the overhead of relaxed scheduling is \emph{independent of the graph size or structure}. 
This analytical result suggests that relaxed schedulers should be a viable alternative, a finding which is also supported by our preliminary concurrent implementation.

Specifically, we consider the following framework. 
Given an input, e.g., a graph, the sequential algorithm defines a set of \emph{tasks}, e.g. one per graph vertex, which should be processed in order, respecting some fixed, arbitrary data dependencies, which can be specified as a DAG. 
Tasks will be accessible via a \emph{scheduler}, which is \emph{relaxed}, in the sense that it could return tasks \emph{out of order}.  
This induces a sequential model,\footnote{We consider this sequential model, similar to~\cite{BFS12}, since there currently are no precise ways to model the contention experienced by concurrent threads on the scheduler. Instead, we validate our findings via a fully concurrent implementation.}  where at each step, the scheduler returns a new task: for simplicity, assume for now  that the scheduler returns at each step a task chosen \emph{uniformly at random} among the top-$k$ available tasks, in descending priority order. (We will model realistic relaxed schedulers~\cite{MQ, AKLN17} precisely in the following section.)
 
%

Assume a thread receives a task from the scheduler. Crucially, the thread \emph{cannot process the task} if it has data dependencies \emph{on higher-priority tasks}: this way, determinism is enforced. (We call such a failed delete attempt by the thread a \emph{wasted} step.) 
However, threads are free to process tasks which do not have such outstanding dependencies, potentially out-of-order (we call these \emph{successful} steps.)
We measure \emph{work} in terms of the total number of scheduler queries needed to process the entire input, including both successful and unsuccessful removal steps. 
 
We provide a simple yet general approach to analyze this relaxed scheduling process, by characterizing the interaction between the dependency structure induced by a given problem on an arbitrary input, and the relaxation factor $k$ in the scheduling mechanism, which yields bounds on expected work when executing such algorithms via relaxed schedulers. 
 Our approach extends to general iterative algorithms, as long as task dependencies are \emph{explicit}, i.e., can be statically expressed given the input, and tasks can be randomly permuted initially.  
 
 The work efficiency of this framework will critically depend on the rate at which threads are able to successfully remove dependency-free tasks. 
Intuitively, this rate appears to be highly dependent on (1) the problem definition, (2) the scheduler relaxation factor $k$, but also on (3) the structure of the input. 
Indeed, we show that in the most general case, a $k$-relaxed scheduler can process an input described by a dependency graph $G$ on $n$ nodes and $m$ edges and incur  $O( \frac{m}{n}\poly(k) )$ wasted steps, i.e. $n + O( \frac{m}{n}\poly(k) )$ total steps. This result immediately implies a low ``cost of relaxation'' for problems whose dependency graph is inherently \emph{sparse}, such as Greedy Coloring on sparse graphs, Knuth Shuffle and List Contraction, which are characterized by a dependency structure with only $m = O(n)$ edges. Hence, in general, such sparse problems incur negligible relaxation cost when $k \ll n$.  

Our main technical result is a counter-intuitive bound for greedy MIS: our framework equipped with a $k$-relaxed scheduler can execute greedy MIS on \emph{any} graph $G$ and experience only $\poly(k)$ wasted steps (i.e. $n+\poly(k)$ total steps), \emph{regardless of the size or structure of $G$}. This result is surprising as well as technically non-trivial, and demonstrates that for MIS on large graphs, operation-level speedups provided by relaxation come with a negligible global trade-off. A similar result holds for maximal matching. 

In the broader context of the parallel scheduling literature, our results suggest that \emph{task priorities} can be supported in a scalable manner, through relaxation, without loss of determinism or work efficiency. We believe this is the first time this observation is made. 
We validate our results empirically, via a preliminary implementation of the scheduling framework in C++, based on a lock-free extension of the MultiQueue relaxed schedulers~\cite{MQ}.
Our broad finding is that this relaxed scheduling framework can ensure scalable execution, with minimal overheads due to contention and verifying task dependencies. 
For MIS on large graphs, we obtain a solution, with 6x speedup at $24$ threads versus an optimized sequential baseline. 

\paragraph{Related Work.} 
Our work is inspired by the line of research by Blelloch et al.~\cite{Blelloch, BFS12, Swarm, blelloch2012internally, BFS14}, as well as~\cite{coppersmith1987parallel, calkin1990probabilistic, FN18}, whose broad goal has been to examine the dependency structure of a wide class of iterative algorithms, 
and to derive efficient scheduling mechanisms given such structure. 

At the same time, there are several differences between these results and our work. 
First, at the conceptual level,~\cite{BFS12, BFS14} start from analytical insights about the dependency structure of algorithms such as greedy MIS, 
and apply them to design scheduling mechanisms which can leverage this structure, which require problem-specific information. 
In some cases, e.g.~\cite{BFS12}, the scheduling mechanisms found to perform best in practice differ from the structure of the schedules analyzed. 
By contrast, we start from a realistic model of existing high-performance relaxed schedulers~\cite{MQ},  and show that such schedulers can automatically and efficiently execute a broad set of iterative algorithms. 
Second, at the technical level, the methods we develop are \emph{different}: for instance, the fact that the iterative algorithms we consider have low dependency depth~\cite{Blelloch, BFS12, BFS14} does not actually help our analysis, since a sequential algorithm could have low dependency depth and be inefficiently executable by a relaxed scheduler: the bad case here is when the dependency depth is low (logarithmic), but each ``level" in a breadth-first traversal of the dependency graph has high fanout.  
Specifically, we emphasize that the notion of \emph{prefix} defined in~\cite{BFS12} to simplify analysis is \emph{different} from the set of positions $S$ which can be returned by the relaxed stochastic scheduler: for example, the parallel algorithm in~\cite{BFS12} requires each prefix to be fully processed before being removed, whereas $S$ acts like a sliding window of positions in our case. 
The third difference is in terms of analytic model: references such as~\cite{BFS12} express work bounds in the CRCW PRAM model, whereas we count work in terms of number of tasks processing attempts. Our analysis is  sequential, and we implement our algorithms on a shared memory architecture to demonstrate empirically good performance.

To our knowledge, the first instance of a relaxed scheduler is in work by Karp and Zhang~\cite{KarZha93}, for parallelizing backtracking strategies in a (synchronous) PRAM model.  
This area has recently become extremely active, with several such schedulers (also called relaxed priority queues) being proposed over the past decade, see~\cite{LotanShavit, Basin11, klsm, SprayList, Haas, Nguyen13, MQ, AKLN17, sagonas2017contention} for recent examples. 
In particular, we note that state-of-the-art packages for graph processing~\cite{Nguyen13} and machine learning~\cite{gonzalez2012powergraph} implement such relaxed schedulers. 

Recent work by a subset of the authors~\cite{AKLN17} showed that a simple and popular priority scheduling mechanism called the MultiQueue~\cite{MQ, Haas, gonzalez2012powergraph} enforces strong probabilistic guarantees on the rank of elements returned, in an idealized model. 
Concurrent work~\cite{ABKLN} proves that these guarantees in fact hold in asynchronous concurrent executions, under some analytic assumptions. 
Based on this result, our work bounds should hold when using MultiQueues as relaxed schedulers, in concurrent executions. 
 
Parallel scheduling~\cite{blumofe, Cilk} is an extremely vast area and a complete survey is beyond our scope. We do wish to emphasize that standard work-stealing schedulers \emph{will not} provide this type of work bounds, since they do not provide any guarantees in terms of the \emph{rank of elements removed}:  the rank becomes unbounded over long executions, since a single random queue is sampled at every stealing step~\cite{AKLN17}. To our knowledge, there is only one previous attempt to add priorities to work-stealing schedulers~\cite{imam2015load}, using a multi-level global queue of tasks, partitioned by priority. This technique is different, and provides no work guarantees. 

\section{Executing Iterative Algorithms via Priority Schedulers}

\subsection{Modeling Relaxed Priority Schedulers}
\label{sec:relax}

In the following, we will provide the sequential specification of a generic relaxed priority scheduler $Q$, which contains a set of $\langle\id{task}, \id{priority}\rangle$ pairs. A relaxed priority scheduler will provide the following methods:
\begin{itemize}
	\item \lit{ApproxGetMin}(), which returns a $\tup{\id{task}, \id{priority}}$ pair and deletes it from the structure, if a task is available, or $\bot$, otherwise. The relaxation guarantees of this operation are precisely defined below;
	\item \lit{Empty}(), which returns whether the scheduler still has tasks or not;       
	\item \lit{Insert}( $\tup{\id{task}, \id{priority}}$ ), which inserts a new task into $Q$. 
\end{itemize}

Let \lit{rank}$(t)$ be the rank of the task which is returned by the $t^{th}$ \lit{ApproxGetMin} operation, among all tasks present in $Q$. 
We say that task $u \in Q$ experiences a \emph{priority inversion} at an \lit{ApproxGetMin} step if a task $v$ of \emph{lower} priority than $u$ is retrieved at that step. For any task $u$, let $\lit{inv}(u)$ be the number of inversions which the task $u$ experiences before being removed. 

\begin{definition}
	Fix a relaxed priority scheduler $Q$, with parameters $k \geq 1$, the \emph{rank bound}, and $\phi$, the \emph{fairness bound}. We say that $Q$ is an $(k, \phi)$-relaxed priority scheduler if it ensures the following:
\begin{enumerate}
	\item \textbf{Rank Bound.} For any time $t$, and any integer $\ell >1$, \\ $\Pr [ \lit{rank} (t) \geq \ell ] \leq \exp\left( - \ell / k \right)$.
	\item \textbf{Fairness Bound.} For any task $u$, and any integer $\ell \geq 1$, $\Pr [ \lit{inv}(u) \geq \ell ] \leq \exp\left( - \ell / \phi \right)$.
\end{enumerate}
\end{definition}

\paragraph{Relation to Practical Schedulers.} 
Upon inspection, both the SprayList~\cite{SprayList} and the MultiQueue~\cite{MQ} relaxed priority schedulers ensure these exponential tail bounds on both rank and fairness, under some analytic assumptions. These conditions are trivially ensured by deterministic implementations such as~\cite{klsm}. 
In particular, the SprayList ensures these bounds with parameters $k$ and $\phi$ in $O( p \polylog p )$, where $p$ is the number of processors~\cite{SprayList}. 
MultiQueues ensure these bounds with parameters $k = O( m )$, and $\phi = O( m \log m )$, where $m$ is the number of distinct priority queues~\cite{AKLN17}. This holds even in concurrent executions~\cite{ABKLN}. 

In the following, it will be convenient to assume a single parameter $k$, which upper bounds both the rank and the relaxation parameters. We call the $(k, k)$-relaxed scheduler simply a $k$-relaxed scheduler.

\subsection{A General Scheduling Framework}
\label{sec:algorithms}

We now present our framework for executing task-based sequential programs, whose pseudocode is given in 
Algorithm~\ref{alg:exact}. 
We assume a permutation $\pi$ which dictates an execution order on tasks. 
If $u$ is the $i^{th}$ element in $\pi$, we will write $\pi(i) = u$ and $\ell(u) = i$ ($\ell$ for \emph{label}). 
Algorithm~\ref{alg:exact} encapsulates a large number of common iterative algorithms on graphs, including Greedy Vertex Coloring, Greedy Matching, Greedy Maximal Independent Set, Dijkstra's SSSP algorithm, and even some algorithms which are not graph-based, such as List Contraction and Knuth Shuffle~\cite{Blelloch}. 
We show sample instantiations of the framework in Section~\ref{sec:examples}.

\DontPrintSemicolon

\begin{algorithm}
\caption{Generic Task-based Framework}
\label{alg:exact}
{
\small
    \KwData{\emph{Dependency} Graph $G = (V,E)$}
    \KwData{Vertex permutation $\pi$}
        
    \CommentSty{// Q is an exact priority queue}
    
    $Q \gets$ vertices in $V$ with priorities $\pi(V)$\; 
    
    \For{each step $t$}
    {
        \CommentSty{// Get new element from the buffer}
        
        $v_t \gets Q.\mathsf{GetMin}()$\;
        $\mathsf{Process}(v)$\label{line:exact-process} 
            
        Remove $v_t$ from $Q$\;

        \If{$Q.\mathsf{empty}()$}
        {
            { $\mathsf{break}$\;}
        }
    }
}
\end{algorithm}

\begin{algorithm}
\caption{Relaxed Scheduling Framework}
\label{alg:generic}
{
\small
    \KwData{\emph{Dependency} Graph $G = (V,E)$}
    \KwData{Vertex permutation $\pi$}
    \KwData{Parameter $k$}
        
    \CommentSty{// Q is a $k$-relaxed scheduler}
    
    $Q \gets$ vertices in random order\; 
    
    \For{each step $t$}
    {
        \CommentSty{// Get new element from the buffer}
        
        $v_t \gets Q.\mathsf{ApproxGetMin}()$\;
        \If{ $v_t$ has unprocessed predecessor\label{line:find-pred} }
        { $Q.insert(v_t,\pi(v_t))$\CommentSty{ // Failed; reinsert}\;$\mathsf{continue}$\label{line:generic-continue} }
        \lElse
        {
            $\mathsf{Process}(v)$\label{line:process} 
        }
        
       \lIf{$Q.\mathsf{empty}()$}
        {
            { $\mathsf{break}$}
        }
    }
}
\end{algorithm}

Algorithm~\ref{alg:generic} gives a method for adapting Algorithm~\ref{alg:exact} to use a \emph{relaxed} queue, given an explicit dependency graph $G=(V,E)$ whose nodes are the tasks, and whose edges are dependencies between tasks.
Importantly, given the dependency graph $G$, Algorithm~\ref{alg:generic} gives the same output as Algorithm~\ref{alg:exact}, irrespective of the relaxation factor $k$.  As usual, we write $|V| = n$ and $|E| = m$. We assume that the permutation $\pi$ represents a \emph{priority order} so that an edge $e = (u,v) \in E$ means that $v$ depends on $u$ if $\ell(v) > \ell(u)$ and vice-versa. In the former case, we say that $v$ is a \emph{successor} of $u$ and $u$ is a \emph{predecessor} of $v$.

Our main result regarding Algorithm~\ref{alg:generic}, proven formally in Section~\ref{sec:generic-analysis}, argues that if $\pi$ is chosen uniformly at random from among all vertex permutations, then Algorithm~\ref{alg:generic} completes in at most $n + O(\frac{m}{n}poly(k))$ iterations (compared to exactly $n$ for Algorithm~\ref{alg:exact}). This result demonstrates that provided $G$ is not too dense, the ``cost of relaxation'' is low for the class of problems which admit uniformly random task permutations. Notably, this class includes all of the problems mentioned above, except for Dijkstra's algorithm (since there, $\pi$ needs to respect the ordering of nodes sorted by distance from the source).

\subsection{Example Applications}
\label{sec:examples}

Applying the sequential task-based framework of Algorithm~\ref{alg:exact} only requires an implementation of $\mathsf{Process(v)}$. 
Implementing the relaxed framework in Algorithm~\ref{alg:generic} further requires $G$ (either explicitly or via a predecessor query method). 
We now give examples for Greedy Vertex Coloring and List Contraction, whose dependency graph is implicit.

\paragraph{Greedy Vertex Coloring.} Vertex Coloring is the problem of assigning a \emph{color} (represented by a natural number) to each vertex of the input graph, $G$, such that no adjacent vertices share a color. The Greedy Vertex Coloring algorithm simply processes the vertices in some permutation order, $\pi$, and assigns each vertex in turn the smallest available color. The implementation of $\mathsf{Process(v)}$ for Greedy Vertex Coloring needs to determine the color of $v$, which can be done as described below:

\begin{algorithm}
\caption{Greedy Vertex Coloring $\mathsf{Process}(v)$}
\small
\KwData{Input Graph $G = (V,E)$}
\KwData{Permutation $\pi$}
\KwData{\emph{Partial} coloring $c: V\to \mathbb{N}$}
\Fn{$\mathsf{Process(v):}$}
{
	$S \gets \emptyset$\;
	\ForEach{$(u,v)\in G,$ s.t. $\ell(u) < \ell(v)$}
	{
		$S \gets S \cup \{c(u)\}$\;
	}
	$c(v) \gets \min_{i \in \mathbb{N}} i\notin S$\;
}
\end{algorithm}

Since the underlying dependency graph is just the input graph with edge orientations given by $\pi$, this is all that needs to be provided.

\paragraph{List Contraction.} List Contraction takes a doubly linked list, $L$, and iteratively \emph{contracts} its nodes. Contracting a node $v$ consists of swinging two pointers: $v\mathsf{.next.previous} \gets v\mathsf{.previous}$ and $v\mathsf{.previous.next} \gets v\mathsf{.next}$, effectively removing $v$ from the list. List Contraction is useful, e.g., for cycle counting. Although List Contraction is not inherently a graph problem, we can still construct a dependency graph $G$ whose nodes are list elements and with an edge between elements which are adjacent in $L$. If we induce a priority order on list elements (e.g. uniformly at random), then there is an induced orientation of the edges of the dependency graph which forms a DAG. Then a predecessor query on $v$ consists of checking whether either $v\mathsf{.next}$ or $v\mathsf{.prev}$ is an unprocessed predecessor. $\mathsf{Process}(v)$ can be implemented with just the two steps of contraction above (possibly along with the metrics the application is computing).

\subsection{Greedy Maximal Independent Set}

We give a variant of Algorithm~\ref{alg:generic} adapted for Greedy Maximal Independent Set (MIS), which makes use of some exploitable substructure. In particular, once some neighbor, $u$, of a vertex $v$ is added to the MIS, then $v$ can never be added to the MIS, at which point $v$'s dependents no longer have to wait for $v$ to be processed. Algorithm~\ref{alg:mis} implements MIS in the framework of Algorithm~\ref{alg:generic} while also making use of this observation. Interestingly, Algorithm~\ref{alg:mis} can also be used to find a maximal matching by taking the input graph $G$ of the matching instance and converting it to a graph $G'$, where $G'$ has a vertex for each edge in $G$ and there is an edge between vertices of $G'$ if the corresponding edges of $G$ share an incident vertex. (One can view matching as an ``independent set'' of edges, no two of which are incident to the same vertex.)

\begin{algorithm}
\caption{Relaxed Queue MIS}
\label{alg:mis}
{\scriptsize
    \KwData{Graph $G=(V,E)$}
    \KwData{Vertex permutation $\pi$}
    \KwData{Parameter $k$}
    
    \CommentSty{// Q is a $k$-approximate priority queue}
    
    $Q \gets$ vertices in random order, all marked \emph{live}\; 
    
    \For{each step $t$}
    {
        \CommentSty{// Get new element from the buffer}
        
        $v_t \gets Q.\mathsf{ApproxGetMin}()$\;
        \lIf{ $v_t$ marked dead}
        	{  continue }
        \ElseIf{ $v_t$ has live predecessor in $Q$ }
        { $Q.insert(v_t,\pi(v_t))$\CommentSty{ // Failed; reinsert}\;$\mathsf{continue}$\label{line:mis-reinsert} }
        \Else
        {
            Add $v_t$ to MIS\; \label{line:add}
            Mark all of $v_t$'s neighbors dead\;
            Remove $v_t$ from $Q$\;
        }
        
        \lIf{$Q.\mathsf{empty}()$}
        {
            { $\mathsf{break}$}
        }
    }
}
\end{algorithm}

As we will show in Section~\ref{sec:analysis}, the simple improvement Algorithm~\ref{alg:mis} makes over Algorithm~\ref{alg:generic} results in only a negligible number of extra iterations due to relaxation.

\section{Analysis}
\label{sec:analysis}

In this section, we will bound the relaxation cost for the general framework (Algorithm~\ref{alg:generic}) and for Maximal Independent Set (Algorithm~\ref{alg:mis}). Algorithm~\ref{alg:generic} is easier to analyze and will serve as a warmup. Note that in both cases, $n$ iterations are required to process all nodes and are necessary even with no relaxation. Thus, we can think of the ``cost'' of relaxation as the number of further iterations beyond the first $n$, which can be equivalently counted as the number of \emph{re-insertions} performed by the algorithm. We will sometimes refer to executing such a re-insertion as a ``failed delete'' by $Q$.

Our primary goal will be to bound the number of iterations of the for loops in Algorithm~\ref{alg:generic} and~\ref{alg:mis} when running them sequentially with a $k$-relaxed priority queue. Although the initial analysis is sequential, the algorithms are parallel: threads can each run their own for loops concurrently and correctness is maintained. The difficulty in extending the analysis to the asynchronous setting is that it is not clear how to model failed deletes of dependents of a node that is being processed. The likelihood of such deletes depend on particulars of both the problem (i.e. how long processing and dependency checking steps actually take) and the thread scheduler and so are hard to model in our generic framework. Instead, we show empirically that our bounds hold in practice on a realistic asynchronous machine where threads run the loops fully in parallel. 

The theorems we will prove are the following. Given a dependency graph $G = (V,E)$ with $|V| = n$ vertices and $|E| = m$ edges, we first bound the number of iterations of Algorithm~\ref{alg:generic}:

\begin{theorem}
\label{thm:generic}
Algorithm~\ref{alg:generic} runs for $n + O\left(\frac{m}{n}\right)\poly(k)$ iterations in expectation.
\end{theorem}

By contrast to Algorithm~\ref{alg:generic}, we show that using a relaxed queue for computing Maximal Independent Sets on large graphs has essentially no cost at all, even for dense graphs! In particular, Algorithm~\ref{alg:mis} incurs a relaxation cost with no dependence at all on the size or structure of $G$, only on the relaxation factor $k$:

\begin{theorem}
\label{thm:mis}
Algorithm~\ref{alg:mis} runs for $n + \poly(k)$ iterations in expectation.
\end{theorem}

Before delving into the individual analyses, we first consider some key characteristics of a particular relaxed queue which will be at play, and quantify them in terms of the \emph{fairness} and \emph{rank error} of $Q$. As discussed in Section~\ref{sec:relax}, we will assume that $Q$ is $k$-relaxed: that is, $Q$ provides exponential tail bounds on the rank error and on the number of inversions experienced by an element, in terms of the parameter $k$.   
Intuitively, it may help to think of a queue which returns a uniformly random element of the \topk{} at each step as the ``canonical'' $k$-relaxed $Q$. See Figure~\ref{fig:topk} for an illustration. (As discussed in Section~\ref{sec:relax}, real schedulers have slightly different properties, which are captured in our framework.) 
We state and prove two technical lemmas parameterized by $k$. 

First, we characterize the probability that, for some edge $e = (u,v)$ in the dependency graph where $u$ is a predecessor of $v$, $v$ experiences an inversion before $u$ is processed.
We say that vertex $u$ experiences an inversion on or above node $v$ at some point during the execution if $\ell(u) < \ell(v)$,  but some node with label \emph{at least} $\ell(v)$ is returned by $Q$ before $u$ is processed during the execution. 

\begin{figure}
\centering
	\includegraphics[scale=0.4]{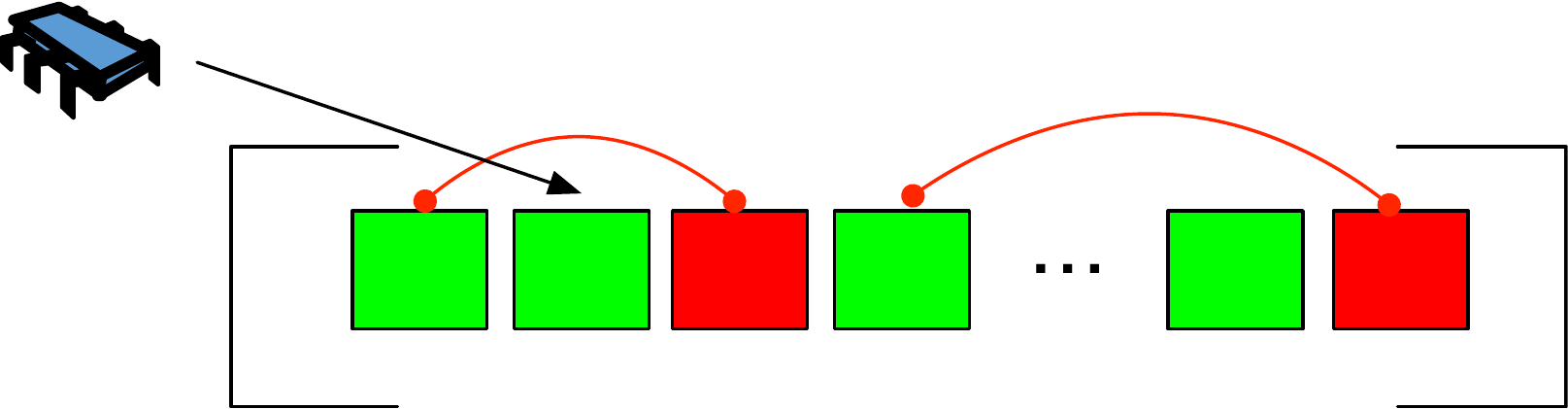}
	\caption{Simple illustration of the process. The blue thread queries the relaxed scheduler, which returns one of the top $k$ tasks, on average (in brackets). Some of these tasks (green) can be processed immediately, as they have no dependencies. Tasks with dependencies (red) cannot be processed yet, and therefore result  in failed deletes. }
\label{fig:topk}
\end{figure}

\begin{lemma}
\label{lem:active}
Consider running Algorithm~\ref{alg:generic} (or Algorithm~\ref{alg:mis}) using a $k$-relaxed queue $Q$ on input graph $G=(V,E)$. For a fixed edge $e = (u,v)$, the probability that 
$u$ experiences an inversion on or above $v$ during the execution is bounded by $O(k^3 \log k/n)$.
\end{lemma}
\begin{proof}
	We begin by proving a few immediate claims.
\begin{claim}
\label{claim:top}
	At any time $t$, the probability of removing the element of top rank from $Q$ is at least $1 / k$. 
\end{claim}
\begin{proof}
	By the rank bound, we have that $\Pr[ \lit{rank} (t)  \geq 2 ] \leq (1 / e)^{2 / k} < 1 - 1 / k.$
	It therefore follows that $\Pr [ rank(t) = 1 ] \geq 1 / k$. 
\end{proof}

Let $t_u$ be the first time when $u$ experiences an inversion, and let $R_u$ be its rank at that time. 
Since an element of rank $\geq u$ must be chosen at $t_u$, we have that, for any $\ell \geq 1$,
$$ \Pr[  R_u \geq \ell  ]  \leq \exp( - \ell / k).$$

\noindent In particular, $\Pr [ R_u \geq ck \log k ] \leq (1 / k)^c$, for any constant $c \geq 1$. 
That is, $u$ has rank $\leq ck \log k $ at the time where it experiences its first inversion, w.h.p. in $k$. 
We now wish to bound the number of removals between the point when $u$ experiences its first inversion, and the point when $u$ is removed. 
Let this random variable be $\Delta_k$. 
By Claim~\ref{claim:top}, the top element is always removed within $O( k )$ trials in expectation, and hence we can show that 
$$\Delta_k \geq 2c k^2 \log k, \textnormal{ with probability at most $2 / k^c$,}$$
\noindent for $c \geq 1$, by using inequality :

\begin{equation}
\Pr[\Delta_k \geq 2c k^2 \log k] \le \Pr[R_u \geq ck \log k]+
\Pr[\Delta_k \geq 2c k^2 \log k | R_u < ck \log k].
\end{equation}

We can prove that $\Pr[\Delta_k \geq 2c k^2 \log k | R_u < ck \log k] \le 1/k^c$ by using bounds for the negative binomial distribution with success probability at least $1/k$, since success in our case is deleting the top element. 

We now wish to know the probability that for a fixed $\Delta_k$ one of these steps is an inversion experienced by $u$ on or  above $v$. Node $v$ has lesser priority than $u$, chosen uniformly at random. 
Let $j$ be the position of $v$, noting that $\Pr[ j = \pi(u) + \ell ] \leq 1/(n-1)\le 2/n$, for any integer $\ell \geq 1$. 
Fix a step $t$, and pessimistically assume that $u$ is at the top of the queue at this time. Note that if $u$ is at the top of the queue, then $v$ has rank larger than $l$ at time step $t$. From the fairness bound we know that $\Pr[inv(u)\ge l/2]\le(1/e)^{l/2k}$. Conditioned on $inv(u) \le l/2$, we have that rank of $v$ is always at least $l/2$ until $u$ is removed. This allows us to use the rank bound for upper bounding the probability of $u$ experiencing inversion on or above $v$ at each step and then we can union bound over $\Delta_k$ steps.
Hence, we get that the probability that $u$ experiences an inversion on or above $v$ is at most $(1/e)^{l/2k}+\Delta_k  (1/e)^{l/2k}$. 
Further, bounding over all choices of $l$, we get the upper bound:
\begin{equation}
\frac{2}{n} \sum_{l=1}^{n-1} \Big( (1/e)^{l/2k}+\Delta_k  (1/e)^{l/2k} \Big)
\le \frac{2}{n}(2k+2k\Delta_k)=\frac{4k}{n}(1+\Delta_k).
\end{equation}

Finally, bounding over all possible values of $\Delta_k$ and their probabilities, we get that 
the probability that $u$ experiences an inversion on or above $v$ during the execution is at most:
\begin{equation}
\sum_{i=0}^{\infty} \frac{2}{k^i}\frac{4k}{n}(1+2(i+1)k^2 \log k)=O(k^3\log k/n).
\end{equation}
\end{proof}

\noindent  Furthermore, the above proof directly implies the following corollary:

\begin{corollary}
\label{cor:active}
Consider running Algorithm~\ref{alg:generic} (or Algorithm~\ref{alg:mis}) using a $k$-relaxed queue $Q$ on input graph $G=(V,E)$. For a fixed edge $e = (u,v)$, the probability that $u$ experiences an inversion on or above $v$ during an execution on a random permutation $\pi$ conditioned on $\ell(u) = t, \ell(v) > t$ is bounded by $O(k^3 \log k/(n-t))$.
\end{corollary}

In Appendix~\ref{app:tighter}, we also prove a slightly tighter version of Lemma~\ref{lem:active} for the case where the implementation of $Q$ provides the further guarantee of \emph{only} returning elements from the \topk. Note that such a queue is always $k$-rank bounded, but is not necessarily $k$-fair.

Our second technical lemma quantifies the expected number of \emph{priority inversions} incurred by an element, $u$, of $Q$ once $u$'s dependencies have been processed---that is, the number of times an element of $Q$ with lower priority than $u$ is returned by $\mathsf{GetApproxMin}()$ before $u$ is. If a vertex $u$ has no predecessor in $Q$ at some time $t$, we call $u$ a \emph{root}. 

\begin{lemma}
\label{lem:inv}
Consider running Algorithm~\ref{alg:generic} (or Algorithm~\ref{alg:mis}) using a $k$-relaxed queue $Q$ on input graph $G=(V,E)$. For a fixed node $u$, if $u$ is a root at some time $t$, at most $O(k)$ other elements of $Q$ with lower priority than $u$ are deleted after $t$ in expectation.
\end{lemma}
\begin{proof}
Follows immediately from the $k$-fairness provided by $Q$.
\end{proof}

We stress that these two lemmas quantify the entire contribution of (the randomness of) the relaxation of $Q$ to the analysis. The major burden of the analysis, particularly for MIS, is instead to manage the interaction between the randomness of the \emph{permutation} $\pi$ (which is not inherently related to the relaxation of $Q$) and the structure of $G$. Equipped with these lemmas, we are ready to do just that.

\subsection{Algorithm~\ref{alg:generic}: The General Case}
\label{sec:generic-analysis}

The following theorem shows that the relaxed queue in Algorithm~\ref{alg:generic} has essentially no cost for sparse dependency graphs with $m = O(n)$ and still completes in $O(nk)$ iterations even for dense dependency graphs when $m = O(n^2)$. For example, Theorem~\ref{thm:generic}  demonstrates that task-based problems which are inherently sparse such as Knuth Shuffle and List Contraction~\cite{Blelloch} incur only negligible ``wasted work'' when utilizing a $k$-relaxed queue with $k \ll n$. Furthermore, graph problems with edge dependencies such as greedy vertex coloring incur a cost proportional to the sparsity of the underlying graph. Although the result is not technically challenging, it is tight up to factors of $k$.

\begin{reptheorem}{thm:generic}
For a dependency graph $G = (V,E)$ with $|V| = n$ vertices and $|E| = m$ edges, Algorithm~\ref{alg:generic} runs for $n + O\left(\frac{m}{n}\right)\poly(k)$ iterations.
\end{reptheorem}

\begin{proof}
We will compute the expected number of failed deletes directly as follows: Whenever a failed delete occurs on a node $w$, charge it to the lexicographically first edge, $e=(u,v)$, for which $u$ and $v$ are both unprocessed and $\ell(v) \leq \ell(w)$ (i.e., with possibly $v=w$). Note that (1) such an edge must exist or else a failed delete could not have occurred, (2) the failed delete must represent a priority inversion on $u$, and (3) $u$ must be a root (because $e$ is lexicographically first).  The first time an edge $e$ is charged, we call $e$ the \emph{active} edge until $u$ is processed. Since $u$ is a root for the duration of $e$'s status as active edge, by Lemma~\ref{lem:inv}, $u$ only experiences $O(k)$ priority inversions in expectation while $e$ is active, which upper bounds the number of failed deletes charged to $e$. 

Let $A_e$ be the event that edge $e=(u,v)$ ever becomes active. $A_e$ can only occur if $u$ experiences an inversion on or above $v$ during the execution, which is bounded by $O(k^3 \log k /n)$ by Lemma~\ref{lem:active}. Thus, the total expected cost of $e$ is at most $\E[c(e)] = \prob{A_e}\E[c(e)\big| A_e] = O(k^4 \log k)/n = \poly(k)/n$. There are $m$ edges so the total cost is $\Theta\left(\frac{m}{n}\right)\poly(k)$ as claimed.
\end{proof}

Briefly, to see that Theorem~\ref{thm:generic} is tight (up to factors of $k$), consider executing a greedy graph coloring problem on a clique. In this case, at any step, only the highest priority node can ever be processed, and for each such node, $u$, it takes $O(k)$ delete attempts before $u$ is processed. Thus in total, the algorithm runs for $O(nk)$ iterations.

\subsection{Algorithm~\ref{alg:mis}: Maximal Independent Set}
\label{sec:mis-analysis}
The following theorem bounds the number of iterations of Algorithm~\ref{alg:mis}.  By contrast to Algorithm~\ref{alg:generic}, we show that using a relaxed queue for computing Maximal Independent Sets on large graphs has essentially no cost at all, even for dense graphs! In particular, Algorithm~\ref{alg:mis} incurs a relaxation cost with no dependence on the size or structure of $G$, only on the relaxation factor $k$. 

\begin{reptheorem}{thm:mis}
Algorithm~\ref{alg:mis} runs for $n + \poly(k)$ iterations.
\end{reptheorem}

\begin{proof}

Denote the lexicographically first MIS of $G$ with respect to $\pi$ as $MIS_\pi$. We first identify the key edges in the execution of Algorithm~\ref{alg:mis}. We will say an edge $e=(u,v)$ is a \emph{hot edge} w.r.t. $\pi$ if $u$ is the smallest labeled neighbor of $v$ in $MIS_\pi$. Note that if $(u,v)$ is a hot edge, $v$ is not in $MIS_\pi$ and $u$ has a smaller label than $v$. Let $H_e$ be the event that $e$ is a hot edge w.r.t. $\pi$. Importantly, $H_e$ depends only on the randomness of $\pi$ and not on the randomness of the relaxation of $Q$. We make two key observations about hot edges that will allow us to prove the theorem:

\begin{claim}
\label{cl:num-forward}
There is exactly one hot edge incident to each vertex $v \in V \setminus MIS_\pi$, and therefore the total number of hot edges is strictly less than $n$.
\end{claim}

This is clear from the condition that $u$ is the smallest labeled neighbor of $v$ in $MIS_\pi$ and the fact that if $v$ is not in the $MIS_\pi$, $v$ must have at least one neighbor in $MIS_\pi$, or else $MIS_\pi$ isn't maximal.

\begin{claim}\label{cl:hot-necessary} A node $w$ is only re-inserted by Algorithm~\ref{alg:mis} if there is at least one hot edge $e = (u,v)$ with $u$ a root and $\ell(w) \geq \ell(v)$ (with possibly $v=w$). If $e$ is such an edge, we say $e$ is \emph{active}. Furthermore, at least one active hot edge satisfies $\ell(u) < \ell(w)$.
\end{claim}

If $w$ is re-inserted, then $w$ must be live and adjacent to some smaller labeled live vertex $u$. Either $u$ is a root, in which case $(u,w)$ is the claimed hot edge, or else $u$ must be adjacent to an even smaller labeled live vertex. In the latter case, we can recurse the argument down to $u$ and eventually find a hot edge. In either case, both nodes incident to the discovered active hot edge will have a label no greater than $w$'s.

\paragraph{Proof Outline.} The strategy from here is a follows: whenever a failed delete occurs on a node $w$, we will charge it to an arbitrary hot edge $e=(u,v)$ with $u$ a root and $\ell(w) \geq \ell(v)$ (of which there must be at least one by Claim~\ref{cl:hot-necessary}). Similar to Theorem~\ref{thm:generic}, we will say that $e$ is active during the interval between the first time $u$ experiences an inversion on or above $v$ and the time $u$ is processed. We say that the \emph{cost}, $c(e)$, of an edge, $e$, is the number of failed deletes charged to it (which is notably $0$ unless $e$ is both hot and, at some point, active). We then separately bound (1) the expected number of active hot edges which ever exist over the execution of Algorithm~\ref{alg:mis} and (2) the expected number of failed deletes charged to an edge, given that it is an active hot edge. Combining these will give the result.

In order to quantify the distribution of hot edges, we will need one more definition. Fix $e=(u,v)$ and let $G_{e}$ be the subgraph of $G$ induced by $V' = V\setminus \{u,v\}$ and let $\pi_{e}$ be $\pi$ restricted to $V'$. Let $L_{e,t}$ be the event that neither $u$ nor $v$ has a neighbor $w\in MIS_{\pi_e}$ with $\ell_{\pi_e}(w) < t$. Informally, $L_{e,t}$ is the event that both $u$ and $v$ are still \emph{live} in $G$ after running Algorithm~\ref{alg:mis} with an exact queue ($k=1$) for $t-1$ iterations but with $u,v$ excluded from $Q$. Like $H_e$, $L_{e,t}$ depends only on $\pi$ and not on the randomness of the relaxation of $Q$; furthermore, $L_{e,t}$ is independent from $\ell(u)$ and $\ell(v)$. 
Using this definition, we can compute:

\begin{align*}
\Pr\left[H_e\right] &= \sum_t \Pr\left[L_{e,t}\right]\Pr\left[\ell(u) = t\right]\Pr\left[\ell(v) > \ell(u) \big| \ell(u) = t\right] \\
&= \sum_t \Pr\left[L_{e,t}\right]\frac{1}{n}\frac{n-t}{n-1}.
\end{align*} 
and
\begin{align*}
\Pr\left[\ell(u) = t \big| H_e\right] &= \frac{\Pr\left[H_e \big| \ell(u) = t\right]\Pr\left[\ell(u)=t\right]}{\Pr\left[H_e\right]}\\
&=\frac{\Pr\left[L_{e,t}\right]\Pr\left[\ell(v) > t \big| \ell(u)=t\right]\Pr\left[\ell(u)=t\right]}{\Pr\left[H_e\right]}\\
&= \frac{\Pr\left[L_{e,t}\right]\frac{n-t}{n-1}\frac{1}{n}}{\sum_{t'} \Pr\left[L_{e,t'}\right]\frac{1}{n}\frac{n-t'}{n-1}}\\ &= \frac{\prob{L_{e,t}}(n-t)}{\sum_{t'} \Pr\left[L_{e,t'}\right](n-t')}.
\end{align*}

Next, we use the above formulations to bound the probability that a hot edge $e$ is ever \emph{active}. Suppose we are given that $e$ is a hot edge and $\ell(u) = t$. Then $e$ becomes active if and only if $u$ suffers an inversion on or above $v$ before $u$ is processed by the algorithm. Let $A_e$ be the event that $e$ becomes active. At this point, we might wish to apply Lemma~\ref{lem:active} directly, but unfortunately it is not clear that $\prob{A_e}$ is independent from $H_e$, which we will need. However, note that $H_e$ entails $\ell(v) > \ell(u)$ but given only that, $\ell(v)$ is otherwise independent from $H_e$. Thus, if we  condition on $\ell(u)=t$ and $\ell(v) > \ell(u)$, then $\ell(u)$ is fixed and $\ell(v)$ is (conditionally) independent from $H_e$, and therefore $A_e$ also becomes (conditionally) independent from $H_e$. Now we can apply Corollary~\ref{cor:active}, giving

\begin{align*}
\prob{A_e\big| \ell(u) = t, H_e} &= \prob{A_e\big| \ell(u) = t, \ell(v) > \ell(u)}\\ 
&= O\left(\frac{k^3 \log k}{n-t}\right).
\end{align*}

Then:

\begin{align*}
\Pr\left[A_e \big| H_e\right] = &\sum_t \Pr\left[\ell(u) = t \big| H_e\right] \prob{A_e \big| \ell(u) = t, H_e}  \\ = & \sum_t \frac{\prob{L_{e,t}}(n-t)}{\sum_{t'} \Pr\left[L_{e,t'}\right](n-t')}\frac{O(k^3 \log k)}{n-t} \\
&= O(k^3 \log k) \frac{\sum_t \prob{L_{e,t}}}{\sum_{t'}\prob{L_{e,t'}}(n-t')}.
\end{align*}

Observe that for fixed $e=(u,v)$, $\prob{L_{e,t}}$ is decreasing in $t$. In particular, for any permutation $\pi$ in which the event $L_{e,t}$ occurs, $L_{e,t-1}$ occurs also, but the reverse is not true. Let $\mu = \frac{1}{n} \sum_t \prob{L_{e,t}}$. Using \emph{Chebyshev's sum inequality}, we obtain:
\begin{align*}
\prob{A_e \big| H_e} &\leq O(k^3 \log k)\frac{n\mu}{\sum_{t'} \mu(n-t')} \\
&= O(k^3 \log k)\frac{n}{\sum_{t'} (n-t')}\\
&= O\left(\frac{k^3 \log k}{n}\right).
\end{align*}

Finally, since $u$ is a root and we only charge $e$ for failed deletes on nodes with a larger label than $v$, and therefore a larger label than $u$ as well, the number of times we charge $e$ is upper bounded by the total number of priority inversions suffered by $u$ while a root, which, by Lemma~\ref{lem:inv}, is given by $O(k)$ in expectation. Thus $\E[c(e) \big| A_e, H_e] = O(k)$. 

Combining all the parts, we have a final bound on the total cost:
\begin{align*}
\E\left[\sum_e c(e)\right] &= \sum_e \prob{H_e}\prob{A_e|H_e}\prob{c(e)|A_e, H_e}\\
&= \sum_e \prob{H_e} \cdot O\left(\frac{k^3 \log k}{n}\right)\cdot O(k)\\
&= O\left(\frac{k^4 \log k}{n}\right)\E\left[\#\{H_e\}\right]\\
&\stackrel{Claim~\ref{cl:num-forward}}{<} O\left(\frac{k^4 \log k}{n}\right)\cdot n\\
& = O(k^4 \log k) = \poly(k), \textnormal{ q.e.d. }
\qedhere
\end{align*}
\end{proof}
\section{Experimental Results}

\paragraph{Synthetic Tests.} 
To validate our analysis, we implemented the sequential relaxed framework described in Algorithm~\ref{alg:generic}, 
and used it to solve instances of MIS, matching, Knuth Shuffle, and List Contraction using a relaxed scheduler which uses the MultiQueue algorithm~\cite{MQ}, for various relaxation factors.
We record the average number of extra relaxations, that is, the number of failed deletes during the entire execution, across five runs. 
Results are presented in Table~\ref{tab:table1}. 
We considered graphs of various densities with $10^3$ and $10^4$ vertices.  
The results appear to confirm our analysis: the number of extra iterations required for MIS is low, and scales only in $K$ and not in  $|V|+|E|$. There is some variation
for fixed $K$ and varying $|V|+|E|$, but it is always within a factor of $2$ for our trials and does not appear to be obviously correlated with $|V|+|E|$.  

\begin{table}
  \begin{center}

\begin{tabular}{|r|r||c|c|c|c|c|}
\cline{3-7}                   
\multicolumn{2}{c|}{} & \multicolumn{5}{c|}{$k$}\\                                          
\hline                                                                                      
\multicolumn{1}{|c|}{$|V|$} & \multicolumn{1}{c||}{$|E|$} & 4 & 8 & 16 & 32 & 64\\                                                      
\hline                                                                                      
\multirow{3}{*}{1000} & 10000 & 12.8 & 56.8 & 148.8 & 308.6 & 583.0\\                       
& 30000 & 7.0 & 40.8 & 108.6 & 264.2 & 478.6\\                                               
& 100000 & 12.4 & 40.0 & 100.6 & 225.8 & 427.2\\           
\hline                                                                                      
\multirow{3}{*}{10000} & 10000 & 11.0 & 43.2 & 145.4 & 336.4 & 738.6\\  
& 30000 & 16.6 & 71.4 & 196.0 & 437.6 & 890.2\\
& 100000 & 13.0 & 56.2 & 144.4 & 290.6 & 529.6\\
\hline
\end{tabular}  

%
%
%
%
     \caption{Simulation results for varying parameters of Maximal Independent Set.  $k$ is the relaxation factor, $n$ is the number of nodes and $m$ is the number of edges. The number of extra iterations is averaged over 2 runs. }
    \label{tab:table1}      
  \end{center}
  
\end{table}

\paragraph{Concurrent Experiments.}
We implemented a simple version of our scheduling framework, using a variant of the MultiQueue~\cite{MQ} relaxed priority queue data structure. 
We assume a setting where the input, that is, the set of tasks, is loaded initially into the scheduler, and is removed by concurrent threads. 
We use lock-free lists to maintain the individual priority queues and we hold pointers to the adjacency lists of each node within the queue elements, 
in order to be able to efficiently check whether a task still has outstanding dependencies. 

We compared to the exact scheduling framework using the Wait-free Queue as Fast as Fetch-and-Add~\cite{wfqueue}. Since there could still be some reordering of tasks due to concurrency,
we elect to use a backoff scheme wherein if an unprocessed predecessor is encountered, we 
wait for the predecessor to process.
In practice this rarely occurs.

\paragraph{Setup.} 
Our experiments were run on an Intel Xeon Gold 6150 machine with 4 sockets, 18 cores per socket and 2 hyperthreads per core, for a total of 36 threads per socket, and 144 threads total.
The machine has 512GB of memory, with 128GB \textit{local} to each socket.
Accesses to local memory are cheaper than accesses to remote memory. 
We pinned threads to avoid unnecessary context switches and to fill up sockets one at a time.
The machine runs Ubuntu 14.04 LTS.
We used the GNU C++ compiler (G++) 6.3.0 with optimization level \texttt{-O3}.

Experiments are performed on $G(n, p)$ random graphs in three classes; {sparse} graphs with $10^8$ nodes and $10^9$ edges, \textit{small dense} graphs with $10^6$ nodes and $10^9$ edges, and \textit{large dense} graphs with $10^7$ nodes and $10^{10}$ edges.
Our experiments were bottlenecked by graph generation and loading time so we were limited to these graph sizes.

For each data point, we run five trials.
In each trial, a graph is generated in parallel by 144 threads, and then we measure the time for $n$ threads to compute an MIS.
Note that, even when $n=1$, we generate the graph using 144 threads.
To ensure that this does not change memory locality in a way that would invalidate our results, we used the \texttt{numactl} utility to cause memory to be allocated locally on \textit{only} the sockets where $n$ threads will compute the MIS.
We verified that this yields the same behavior as experiments where the graph is generated with only $n$ threads. 
(Without using \textit{numactl}, we observed significant slowdowns in the sequential algorithm.)

The number of queues in the MultiQueue is $4\times$ the number of threads.
In our graphs, we plot the average run time on a logarithmic y-axis versus the number of concurrent threads.
Error bars show \textit{minimum} and \textit{maximum} run times.

\begin{figure}
\centering
    \large Sparse graph ($10^8$ nodes and $10^9$ edges)
    
	\includegraphics[scale=0.28]{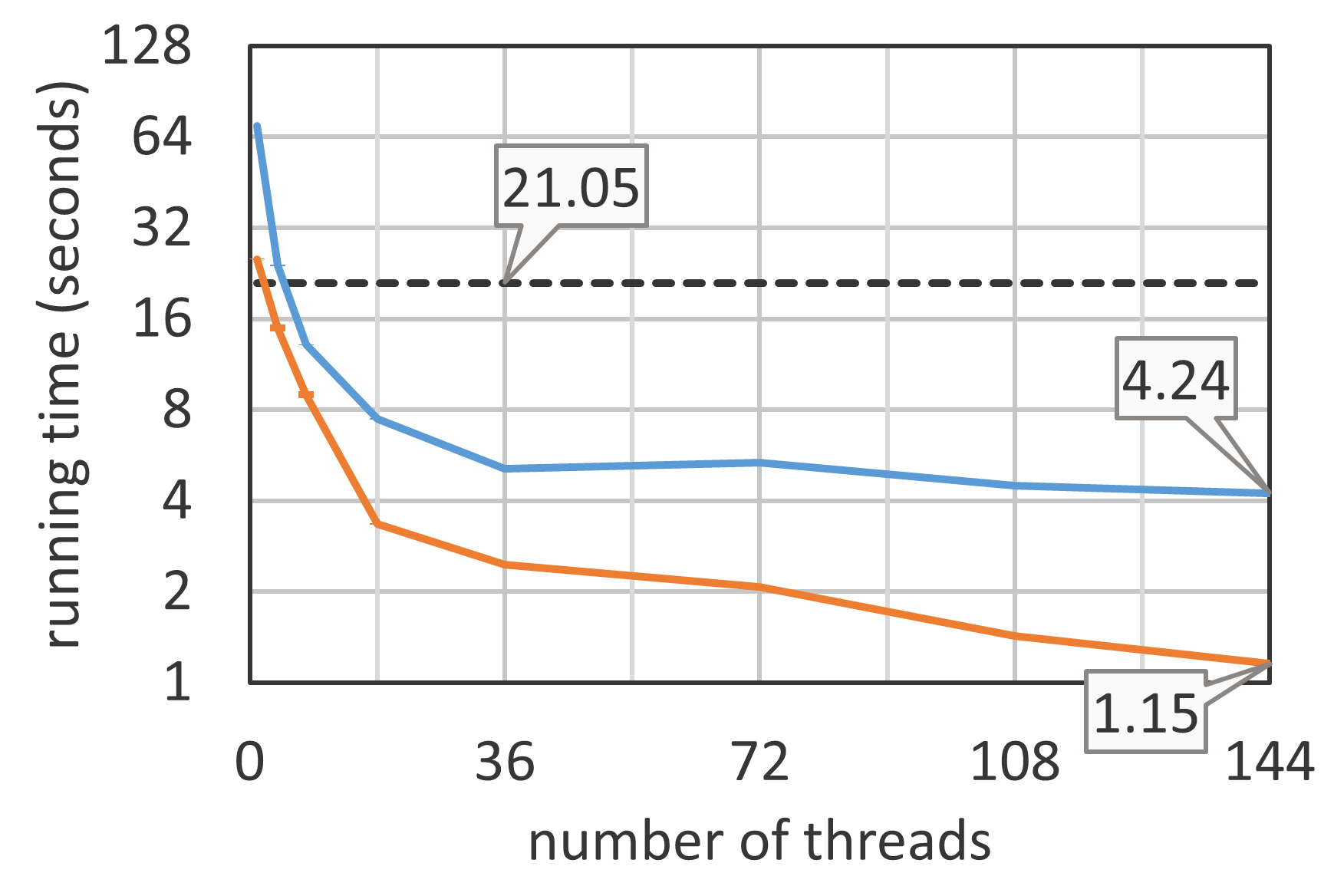}

    \large Small dense graph ($10^6$ nodes and $10^9$ edges)

\includegraphics[scale=0.28]{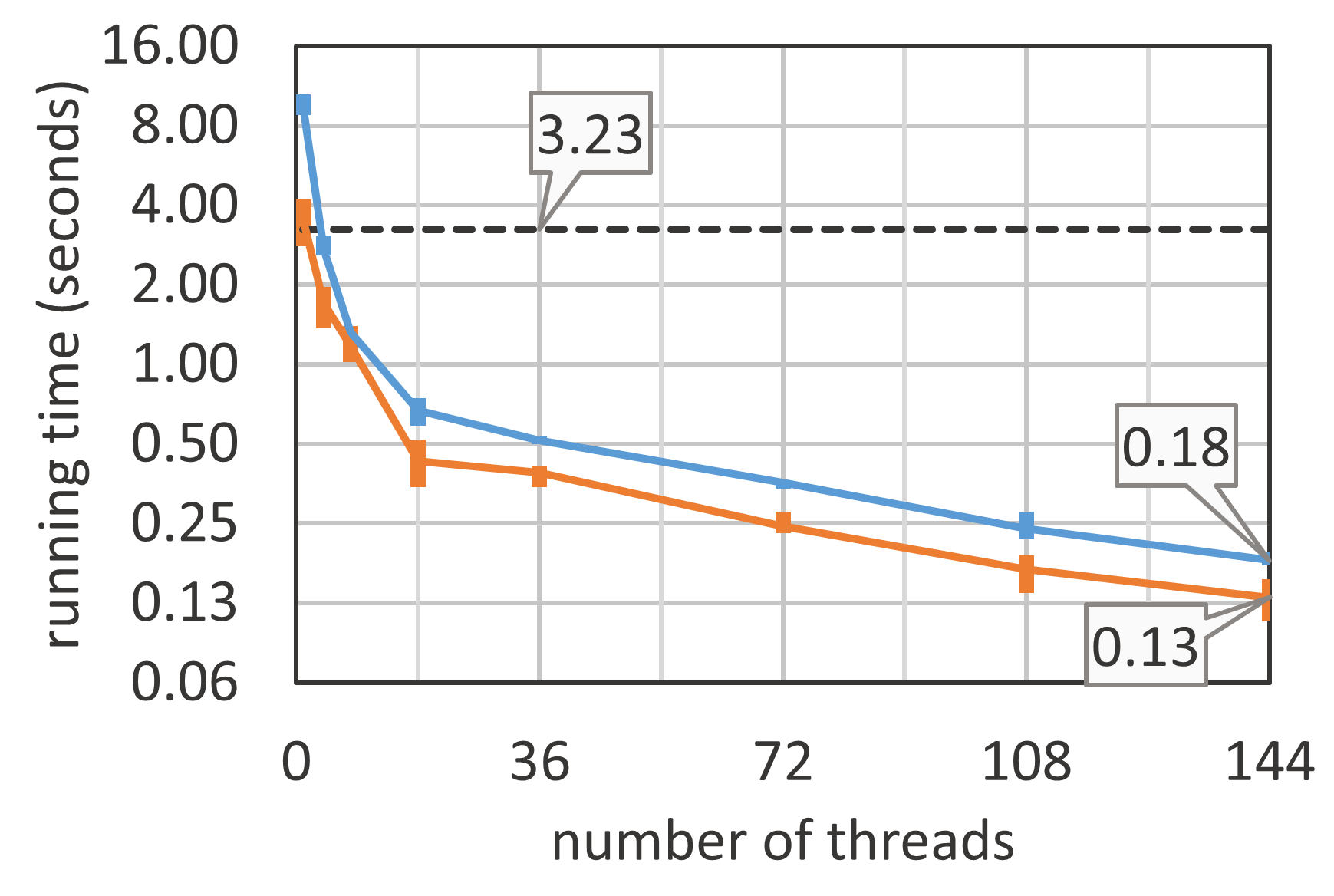}

    \large Large dense graph ($10^7$ nodes and $10^{10}$ edges)
    
	\includegraphics[scale=0.28]{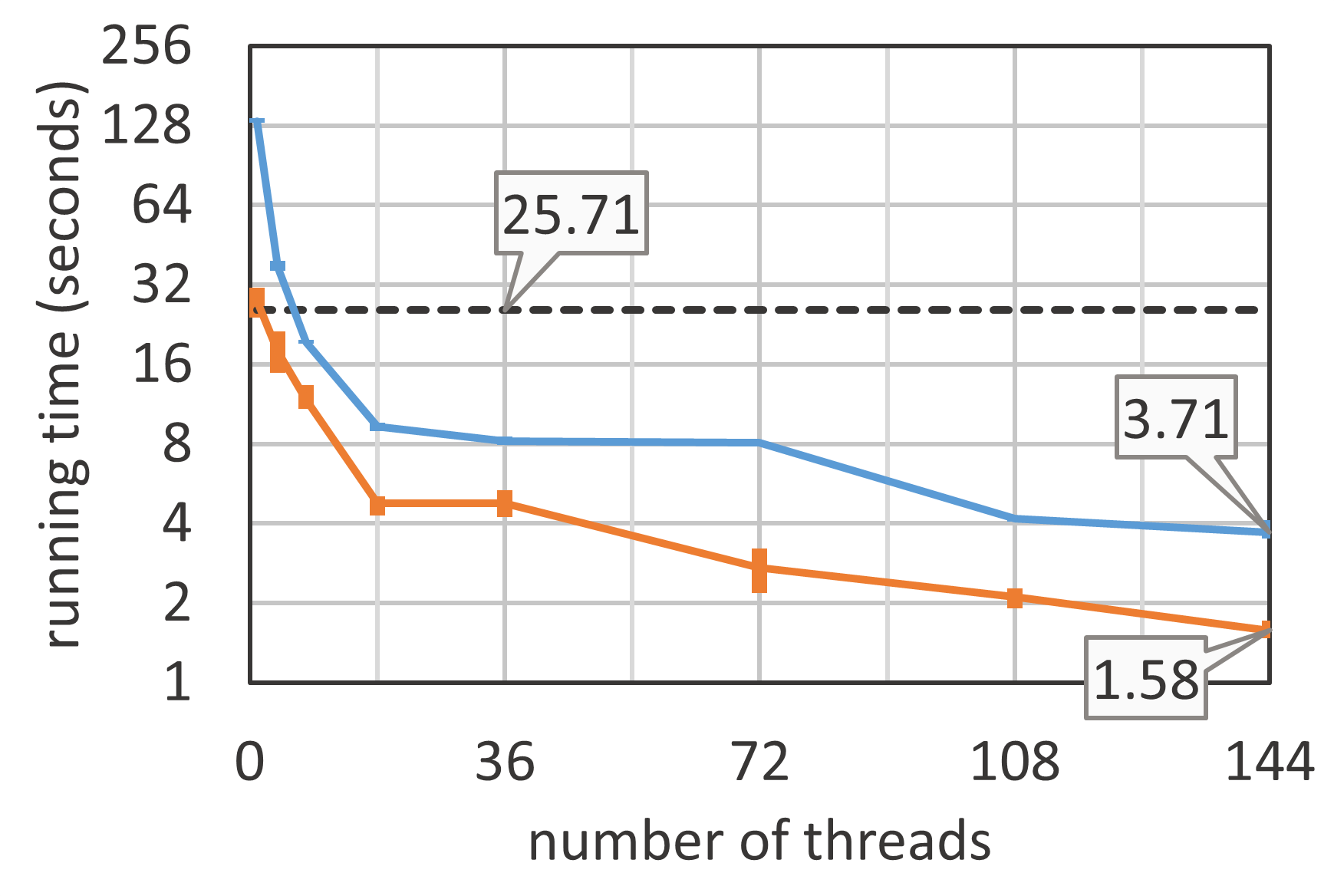}

	\includegraphics[scale=0.3]{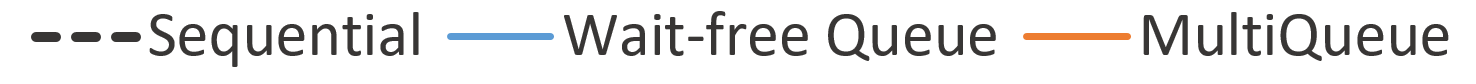}

\caption{Results for concurrent MIS experiments.}
\label{fig:exp_mis}
\end{figure} 

\paragraph{Discussion.} Figure~\ref{fig:exp_mis} shows that our framework using a relaxed scheduler scales with respect to the time to compute MIS over the target graph all the way up to max thread count.
The exact framework using the fast wait-free queue also scales, but not as well. 
In the sparse graphs, the relaxed scheduler is up to $\sim18.2\times$ faster than optimized sequential code, compared to the exact scheduler, which peaks at $\sim 5.0 \times$ faster.
In the small dense graphs, where the time spent traversing edges in the MIS algorithm dominates the minor cost of dequeuing nodes, the exact scheduler achieves a peak speedup of $\sim17.8\times$ over the sequential algorithm, which is approaching the relaxed scheduler's peak speedup of $\sim24.6\times$.
However, in the large dense graphs, even though many edges are still traversed by the MIS algorithm, there are sufficiently many nodes to be dequeued that the performance advantage of the relaxed scheduler shows through: it achieves speedup of up to $\sim16.3\times$ compared to the exact scheduler, which manages only $\sim6.9\times$.
Note that the single threaded performance of the relaxed scheduler is also quite close to the sequential algorithm.
In contrast, the exact scheduler is orders of magnitude slower with a single thread. 

\section{Future Work}
From a theoretical perspective, the natural next step would be to tighten the $\poly(k)$ bound on failed deletes, both for the generic algorithm and for MIS; in fact, we conjecture that the $\poly(k)$ bounds in both Theorems~\ref{thm:generic} and~\ref{thm:mis} can be replaced with $\Theta(k)$. However, proving such a bound seems to require a deep understanding of the interplay between the structure of $G$ and the effects of the randomness of a $k$-relaxed queue, which we had to take care in our analysis to keep \emph{separate}.  Also of interest is to discover more applications, and perhaps more instances like MIS in which the bound in Theorem~\ref{thm:generic} can be improved on. 

One shortcoming of our approach is the fact that our cost measure is the number of \emph{vertex} accesses in the priority queue. 
Notice that in theory our bounds may be substantially different when expressed in other metrics, such as the number of edge accesses for the algorithm to terminate, which are closer to standard \emph{work} bounds. 
We plan to investigate such cost measures in future work.

From the practical perspective, the immediate step would to improve upon our preliminary results, and implement a high-performance variant of this scheduler, and use this framework in the context of more general graph processing packages.

\bibliographystyle{plain}
\bibliography{bibliography}

\appendix
\section{A Tighter Version of Lemma~\ref{lem:active}.} 
\label{app:tighter}
We now prove a slightly tighter version for the case where only elements from the top-k may be chosen by the scheduler $Q$. Note that the condition in the following Lemma (that $u$ and $v$ are simultaneously in the top-k) is a pre-requisite for $u$ to experience an inversion on or above $v$, and thus the Lemma is slightly stronger than necessary.

\begin{lemma}
Consider running Algorithm~\ref{alg:generic} (or Algorithm~\ref{alg:mis}) using a $k$-relaxed queue $Q$ on input graph $G=(V,E)$. For a fixed edge $e = (u,v)$, the probability that both $u$ and $v$ are simultaneously in the top-k of $Q$ during any execution on a random permutation $\pi$ is bounded by $O(k^2/n)$.
\end{lemma}
\begin{proof}
We will write $e\in\topk$ as shorthand for the event that $u$ and $v$ are simultaneously in the top-k of $Q$ at some time. Note that no matter what dependencies exist in the top-k of $Q$,  the entire top-k is flushed after the rank $1$ element gets deleted $k$ times. The number of iterations it takes to delete the rank $1$ element $k$ times after $u$ enters the top-k (thereby flushing $u$ with certainty) is a negative binomially distributed random variable $X_u$ with mean $k^2$ and success probability $1/k$ (due to the fairness of $Q$), and similarly for $X_v$. Since $S_e$ entails that either $\ell(u) < \ell(v) < \ell(u) + X_u$ or $\ell(v) < \ell(u) < \ell(v) + X_v$, we note that the two cases are symmetric and compute:

\begin{align*}
\prob{e\in\topk} &\leq \prob{\ell(u) < \ell(v) < \ell(u) + X_u}\\
&\phantom{=} + \prob{\ell(v) < \ell(u) < \ell(v) + X_v}\\
&= 2\prob{\ell(u) < \ell(v) < \ell(u) + X_u}\\
&= 2\sum_r \prob{X_u = r} \prob{\ell(u) < \ell(v) < \ell(u) + r}\\
&= 2\sum_r \prob{X_u = r} \frac{r}{n}\\
&\leq \frac{2ck^2}{n}\prob{X_u \leq ck^2} + 2\sum_{r > ck^2} \prob{X_u = r}\frac{r}{n}\\
&\leq O\left(\frac{k^2}{n}\right) + 2\sum_{r'} \prob{X_u > r'k^2}\frac{(r'+1)k^2}{n}\\
&\leq O\left(\frac{k^2}{n}\right)\left(1 + \sum_{r'}e^{O(-r')}(r'+1)\right)\tag{*}\\
&=O\left(\frac{k^2}{n}\right),
\end{align*}

where (*) uses a standard tail bound on the Negative Binomial Distribution\footnote{See~\cite{harvey-lecture} for a derivation.}.\end{proof}
%
%
%
%

\end{document}

%% file: settings-custom.tex
\usepackage[ruled, noend]{algorithm2e}
\usepackage{algorithmic}
\usepackage{times}
\usepackage{fullpage}
\usepackage{graphicx}
\usepackage{epsf}
\usepackage{tabularx}
\usepackage{multirow}
\usepackage{caption}
\usepackage{amsmath}
\usepackage{amsthm}
\usepackage{amssymb,latexsym,color}
\usepackage{mdwlist}
\usepackage{url}

\newcommand{\idlow}[1]{\mathord{\mathcode`\-="702D\it #1\mathcode`\-="2200}}
\newcommand{\id}[1]{\ensuremath{\idlow{#1}}}
\newcommand{\litlow}[1]{\mathord{\mathcode`\-="702D\sf #1\mathcode`\-="2200}}
\newcommand{\lit}[1]{\ensuremath{\litlow{#1}}}

\newcommand{\tup}[1]{\langle #1 \rangle}

\newcommand{\E}{\ensuremath{\mathbb{E}}}
\newcommand{\topk}{\text{top-k}}
\newcommand{\prob}[1]{\Pr\left[ #1 \right]}

\newtheorem{lemma}{Lemma}
\newtheorem{theorem}{Theorem}
\newtheorem{claim}{Claim}
\newtheorem{corollary}{Corollary}

\DeclareMathOperator{\poly}{poly}
\newtheorem*{rep@theorem}{\rep@title} \newcommand{\newreptheorem}[2]{%
\newenvironment{rep#1}[1]{%
\def\rep@title{\bf #2 \ref{##1}}%
\begin{rep@theorem} }%
{\end{rep@theorem} } }
\makeatother
\newreptheorem{theorem}{Theorem}
\newreptheorem{lemma}{Lemma}
\newreptheorem{corollary}{Corollary}

\newenvironment{proofT}[1]{\proof\def\toto{#1}}{\hspace*{\fill}$\Box_{Theorem~\ref{\toto}}$\par\vspace{3mm}}

\newenvironment{proofL}[1]{\proof\def\toto{#1}}{\hspace*{\fill}$\Box_{Lemma~\ref{\toto}}$\par\vspace{3mm}}
\newenvironment{proofC}{\proof}{\hspace*{\fill}$\Box_{Corollary~\ref{\toto}}$\par\vspace{3mm}}

\newcommand{\thmpostponed}[2]
{
\newcounter{#1}
\setcounter{#1}{\value{theorem}}
\begin{theorem}
\label{thm:#1}
#2
\end{theorem}
\expandafter\def\csname #1\endcsname{
\newcounter{#1temp}
\setcounter{#1temp}{\value{theorem}}
\setcounter{theorem}{\value{#1}}
\begin{theorem}
#2
\end{theorem}
\setcounter{theorem}{\value{#1temp}}
}
\vspace{-1.5em}
}

\newcommand{\thmpostponedwname}[3]
{
\newcounter{#1}
\setcounter{#1}{\value{theorem}}
\begin{theorem}[#2]
\label{thm:#1}
#3
\end{theorem}
\expandafter\def\csname #1\endcsname{
\newcounter{#1temp}
\setcounter{#1temp}{\value{theorem}}
\setcounter{theorem}{\value{#1}}
\begin{theorem}[#2]
#3
\end{theorem}
\setcounter{theorem}{\value{#1temp}}
}
\vspace{-1.5em}
}

\newcommand{\lemmaproofpostponedwname}[4]
{
\newcounter{#1}
\newcounter{#1temp}
\setcounter{#1}{\value{lemma}}
\begin{lemma}[#2]
\label{#1}
#3
\end{lemma}
\expandafter\def\csname #1\endcsname{
\setcounter{#1temp}{\value{lemma}}
\setcounter{lemma}{\value{#1}}
\begin{lemma}
#3
\end{lemma}
\setcounter{theorem}{\value{#1temp}}
\begin{proofL}{#1}
#4
\end{proofL}
}
}

\newcommand{\thmproofpostponedwname}[4]
{
\newcounter{#1}
\newcounter{#1temp}
\setcounter{#1}{\value{theorem}}
\begin{theorem}[#2]
\label{#1}
#3
\end{theorem}
\expandafter\def\csname #1\endcsname{
\setcounter{#1temp}{\value{theorem}}
\setcounter{theorem}{\value{#1}}
\begin{theorem}
#3
\end{theorem}
\setcounter{theorem}{\value{#1temp}}
\begin{proofT}{#1}
#4
\end{proofT}
}
}

\newcommand{\lemmapostponed}[2]
{
\newcounter{#1}
\newcounter{#1temp}
\setcounter{#1}{\value{theorem}}
\begin{lemma}
\label{#1}
#2
\end{lemma}
\expandafter\def\csname #1\endcsname{
\setcounter{#1temp}{\value{theorem}}
\setcounter{theorem}{\value{#1}}
\begin{lemma}
#2
\end{lemma}
\setcounter{theorem}{\value{#1temp}}
}
\vspace{-1.5em}
}

\newcommand{\lemmapostponedwname}[3]
{
\newcounter{#1}
\newcounter{#1temp}
\setcounter{#1}{\value{theorem}}
\begin{lemma}[#2]
\label{#1}
#3
\end{lemma}
\expandafter\def\csname #1\endcsname{
\setcounter{#1temp}{\value{theorem}}
\setcounter{theorem}{\value{#1}}
\begin{lemma}[#2]
#3
\end{lemma}
\setcounter{theorem}{\value{#1temp}}
}
\vspace{-1.5em}
}

\newcommand{\lemmaproofpostponed}[3]
{
\newcounter{#1}
\newcounter{#1temp}
\setcounter{#1}{\value{theorem}}
\begin{lemma}
\label{#1}
#2
\end{lemma}
\expandafter\def\csname #1\endcsname{
\setcounter{#1temp}{\value{theorem}}
\setcounter{theorem}{\value{#1}}
\begin{lemma}
#2
\end{lemma}
\setcounter{theorem}{\value{#1temp}}
\begin{proofL}{#1}
#3
\end{proofL}
}
}

\newcommand{\corollaryproofpostponed}[3]
{
\newcounter{#1}
\newcounter{#1temp}
\setcounter{#1}{\value{theorem}}
\begin{corollary}
\label{#1}
#2
\end{corollary}
\expandafter\def\csname #1\endcsname{
\setcounter{#1temp}{\value{theorem}}
\setcounter{theorem}{\value{#1}}
\begin{corollary}
#2
\end{corollary}
\setcounter{theorem}{\value{#1temp}}
\begin{proofC}
#3
\renewcommand{\toto}{#1}
\end{proofC}
}
}

\renewcommand{\paragraph}[1]{\vspace{0.1cm} \noindent\textbf{#1}\hspace{0.15em}}

\newcommand{\toto}{xxx}

\SetFuncSty{textsf}
\SetFuncSty{small}

\newtheorem{definition}{Definition} 

\newcommand{\polylog}{\,\textnormal{polylog}\,}

\SetKwProg{Fn}{Function}{}

%% file: mis_arxiv_main.bbl
\begin{thebibliography}{10}

\bibitem{ABKLN}
Dan Alistarh, Trevor Brown, Justin Kopinsky, Jerry Li, and Giorgi Nadiradze.
\newblock Distributionally linearlizable data structures.
\newblock {\em Under Submission to SPAA 2018}, 2018.

\bibitem{AKLN17}
Dan Alistarh, Justin Kopinsky, Jerry Li, and Giorgi Nadiradze.
\newblock The power of choice in priority scheduling.
\newblock In {\em Proceedings of the ACM Symposium on Principles of Distributed
  Computing}, PODC '17, pages 283--292, New York, NY, USA, 2017. ACM.

\bibitem{SprayList}
Dan Alistarh, Justin Kopinsky, Jerry Li, and Nir Shavit.
\newblock The spraylist: A scalable relaxed priority queue.
\newblock In {\em 20th ACM SIGPLAN Symposium on Principles and Practice of
  Parallel Programming}, PPoPP 2015, San Francisco, CA, USA, 2015. ACM.

\bibitem{Basin11}
Dmitry Basin, Rui Fan, Idit Keidar, Ofer Kiselov, and Dmitri Perelman.
\newblock Caf\'{E}: Scalable task pools with adjustable fairness and
  contention.
\newblock In {\em Proceedings of the 25th International Conference on
  Distributed Computing}, DISC'11, pages 475--488, Berlin, Heidelberg, 2011.
  Springer-Verlag.

\bibitem{Blelloch}
Guy~E Blelloch.
\newblock Some sequential algorithms are almost always parallel.
\newblock In {\em Proceedings of the 29th ACM Symposium on Parallelism in
  Algorithms and Architectures, SPAA}, pages 24--26, 2017.

\bibitem{blelloch2012internally}
Guy~E Blelloch, Jeremy~T Fineman, Phillip~B Gibbons, and Julian Shun.
\newblock Internally deterministic parallel algorithms can be fast.
\newblock In {\em ACM SIGPLAN Notices}, volume~47, pages 181--192. ACM, 2012.

\bibitem{BFS12}
Guy~E Blelloch, Jeremy~T Fineman, and Julian Shun.
\newblock Greedy sequential maximal independent set and matching are parallel
  on average.
\newblock In {\em Proceedings of the twenty-fourth annual ACM symposium on
  Parallelism in algorithms and architectures}, pages 308--317. ACM, 2012.

\bibitem{blelloch2016parallelism}
Guy~E Blelloch, Yan Gu, Julian Shun, and Yihan Sun.
\newblock Parallelism in randomized incremental algorithms.
\newblock In {\em Proceedings of the 28th ACM Symposium on Parallelism in
  Algorithms and Architectures}, pages 467--478. ACM, 2016.

\bibitem{Cilk}
Robert~D Blumofe, Christopher~F Joerg, Bradley~C Kuszmaul, Charles~E Leiserson,
  Keith~H Randall, and Yuli Zhou.
\newblock Cilk: An efficient multithreaded runtime system.
\newblock {\em Journal of parallel and distributed computing}, 37(1):55--69,
  1996.

\bibitem{blumofe}
Robert~D Blumofe and Charles~E Leiserson.
\newblock Scheduling multithreaded computations by work stealing.
\newblock {\em Journal of the ACM (JACM)}, 46(5):720--748, 1999.

\bibitem{calkin1990probabilistic}
Neil Calkin and Alan Frieze.
\newblock Probabilistic analysis of a parallel algorithm for finding maximal
  independent sets.
\newblock {\em Random Structures \& Algorithms}, 1(1):39--50, 1990.

\bibitem{coppersmith1987parallel}
Don Coppersmith, Prabhakar Raghavan, and Martin Tompa.
\newblock Parallel graph algorithms that are efficient on average.
\newblock In {\em Foundations of Computer Science, 1987., 28th Annual Symposium
  on}, pages 260--269. IEEE, 1987.

\bibitem{FN18}
Manuela Fischer and Andreas Noever.
\newblock Tight analysis of parallel randomized greedy mis.
\newblock In {\em Proceedings of the Twenty-Ninth Annual ACM-SIAM Symposium on
  Discrete Algorithms}, pages 2152--2160. SIAM, 2018.

\bibitem{gonzalez2012powergraph}
Joseph~E Gonzalez, Yucheng Low, Haijie Gu, Danny Bickson, and Carlos Guestrin.
\newblock Powergraph: distributed graph-parallel computation on natural graphs.
\newblock In {\em OSDI}, volume~12, page~2, 2012.

\bibitem{Haas}
Andreas Haas, Michael Lippautz, Thomas~A. Henzinger, Hannes Payer, Ana
  Sokolova, Christoph~M. Kirsch, and Ali Sezgin.
\newblock Distributed queues in shared memory: multicore performance and
  scalability through quantitative relaxation.
\newblock In {\em Computing Frontiers Conference, CF'13, Ischia, Italy, May 14
  - 16, 2013}, pages 17:1--17:9, 2013.

\bibitem{harvey-lecture}
Nick Harvey.
\newblock Lecture notes in randomized algorithms.
\newblock \url{http://www.cs.ubc.ca/~nickhar/W12/Lecture3Notes.pdf}, 2012.

\bibitem{imam2015load}
Shams Imam and Vivek Sarkar.
\newblock Load balancing prioritized tasks via work-stealing.
\newblock In {\em European Conference on Parallel Processing}, pages 222--234.
  Springer, 2015.

\bibitem{Swarm}
Mark~C Jeffrey, Suvinay Subramanian, Cong Yan, Joel Emer, and Daniel Sanchez.
\newblock Unlocking ordered parallelism with the swarm architecture.
\newblock {\em IEEE Micro}, 36(3):105--117, 2016.

\bibitem{KarZha93}
R.~M. Karp and Y.~Zhang.
\newblock Parallel algorithms for backtrack search and branch-and-bound.
\newblock {\em Journal of the ACM}, 40(3):765--789, 1993.

\bibitem{Nguyen13}
Donald Nguyen, Andrew Lenharth, and Keshav Pingali.
\newblock A lightweight infrastructure for graph analytics.
\newblock In {\em Proceedings of the Twenty-Fourth ACM Symposium on Operating
  Systems Principles}, SOSP '13, pages 456--471, New York, NY, USA, 2013. ACM.

\bibitem{MQ}
Hamza Rihani, Peter Sanders, and Roman Dementiev.
\newblock Brief announcement: Multiqueues: Simple relaxed concurrent priority
  queues.
\newblock In {\em Proceedings of the 27th ACM Symposium on Parallelism in
  Algorithms and Architectures}, SPAA '15, pages 80--82, New York, NY, USA,
  2015. ACM.

\bibitem{making-kjell-happy}
Konstantinos Sagonas and Kjell Winblad.
\newblock The contention avoiding concurrent priority queue.
\newblock In {\em International Workshop on Languages and Compilers for
  Parallel Computing}, pages 314--330. Springer, 2016.

\bibitem{sagonas2017contention}
Konstantinos Sagonas and Kjell Winblad.
\newblock A contention adapting approach to concurrent ordered sets.
\newblock {\em Journal of Parallel and Distributed Computing}, 2017.

\bibitem{LotanShavit}
Nir Shavit and Itay Lotan.
\newblock Skiplist-based concurrent priority queues.
\newblock In {\em Parallel and Distributed Processing Symposium, 2000. IPDPS
  2000. Proceedings. 14th International}, pages 263--268. IEEE, 2000.

\bibitem{BFS14}
Julian Shun, Yan Gu, Guy~E Blelloch, Jeremy~T Fineman, and Phillip~B Gibbons.
\newblock Sequential random permutation, list contraction and tree contraction
  are highly parallel.
\newblock In {\em Proceedings of the twenty-sixth annual ACM-SIAM symposium on
  Discrete algorithms}, pages 431--448. SIAM, 2014.

\bibitem{klsm}
Martin Wimmer, Jakob Gruber, Jesper~Larsson Tr{\"a}ff, and Philippas Tsigas.
\newblock The lock-free k-lsm relaxed priority queue.
\newblock In {\em ACM SIGPLAN Notices}, volume~50, pages 277--278. ACM, 2015.

\bibitem{wfqueue}
Chaoran Yang and John Mellor-Crummey.
\newblock A wait-free queue as fast as fetch-and-add.
\newblock {\em SIGPLAN Not.}, 51(8):16:1--16:13, February 2016.

\end{thebibliography}
